\documentclass[12pt]{article}
\usepackage{amsmath}
\usepackage{graphicx,psfrag,epsf}
\usepackage{enumerate}
\usepackage{psfrag,epsf}
\usepackage{url} 
\usepackage{amsmath,amssymb,amsfonts,amsthm,mathtools,mathrsfs}

\usepackage{bm,bbm}
\usepackage{color}
\usepackage{subfigure}
\usepackage{algorithm,algpseudocode}
\usepackage[symbol]{footmisc}
\usepackage{scalefnt}
\usepackage{authblk} 
\usepackage{multirow,centernot}
\usepackage[colorlinks=true, citecolor=blue, urlcolor=blue]{hyperref}
\usepackage{natbib}
\usepackage{dsfont}
\usepackage[title]{appendix}
\input cyracc.def

\usepackage[left=1in,top=1.1in,right=0.5in,bottom=1in]{geometry}

\theoremstyle{definition}

\newtheorem{theorem}{Theorem}[section]

\newtheorem{proposition}[theorem]{Proposition}
\newtheorem{remark}[theorem]{Remark}

\makeatletter
\def\@seccntformat#1{\@ifundefined{#1@cntformat}%
	{\csname the#1\endcsname\quad}
	{\csname #1@cntformat\endcsname}
}
\makeatother

\newif\ifShowComments
\ShowCommentstrue
\def\strutdepth{\dp\strutbox}
\def\druk#1{\strut\vadjust{\kern-\strutdepth
        {\vtop to \strutdepth{%
                \baselineskip\strutdepth\vss
                        \llap{\hbox{#1}\quad}\null}}}}

\title{\bf
Moment-type estimators for a weighted exponential family 
}

\author{
\text{Roberto Vila}$^{1}$\thanks{Corresponding author: Roberto Vila, email: {rovig161@gmail.com}
\newline
}\,\,\,and
\text{Helton Saulo}$^{1,2}$\\
{\small $^{1}$ Department of Statistics, University of Brasilia, Brasilia, Brazil}\\
{\small $^{2}$ Department of Economics, Federal University of Pelotas, Pelotas, Brazil}\\
}
\begin{document}
	\maketitle 	
	\begin{abstract}
In this paper, we propose and study closed-form moment type estimators for a weighted exponential family. We also develop a bias-reduced version of these proposed closed-form estimators using bootstrap techniques. The estimators are evaluated using Monte Carlo simulation. This shows favourable results for the proposed bootstrap bias-reduced estimators.
	\end{abstract}
	
	
	\section{Introduction}
	\noindent


%
In this work we provide closed-form estimators for the parameters of probability distributions that belongs to the following weighted exponential family \citep{Vila2024b}:
%
%
%
\begin{align}\label{pdf-1}
f(x;\psi)
=
{(\mu\sigma)^{\mu+1} \over (\sigma+\delta_{ab})\Gamma(\mu+1)}\,
[1+\delta_{ab}T(x)]\,
{\vert T'(x)\vert\over T(x)}\,
\exp\left\{-\mu \sigma T(x)+\mu\log(T(x))\right\},
\quad 
x\in (0,\infty),
\end{align}
where $\psi=(\mu,\sigma)$, $\mu,\sigma>0$, $T:(0,\infty)\to (0,\infty)$ is a real strictly monotone twice differentiable function, $\delta_{ab}$ is the Kronecker delta function and $T'(x)$ denotes the derivative of $T(x)$ with respect to $x$.

The probability function $f(x;\psi)$ in \eqref{pdf-1} can be interpreted as a mixture of two distributions that belongs to the exponential family, that is,
\begin{align}\label{decomp-1}
	f(x;\psi)&=
	{\sigma\over \sigma+\delta_{ab}}\,
	f_1(x)
	+
	{\delta_{ab}\over \sigma+\delta_{ab}}\,
	f_2(x),
	\quad x\in (0,\infty), \ \mu, \sigma>0,
\end{align}
where
\begin{align}\label{def-fi}
	f_j(x)
	=
	{(\mu\sigma)^{\mu+j-1} \over \Gamma(\mu+j-1)}\,
	{\vert T'(x)\vert\over T(x)}\,
	\exp\left\{-\mu \sigma T(x)+(\mu+j-1)\log(T(x))\right\},
	\quad 
	x\in (0,\infty), \quad j=1,2.
\end{align}
Densities of form $f_j(x)$, $j=1,2$, have appeared in \cite{Nascimento2014} and \cite{Vila2024,Vila2024b}.

If $X$ has density in \eqref{pdf-1}, from \eqref{decomp-1} and \eqref{def-fi} it is simple to show that the random variable $X$ defined as
\begin{align}\label{rep-st-1}
	X\equiv(1-B)T^{-1}(Z_1)+BT^{-1}(Z_2),
\end{align}
has density in \eqref{pdf-1}, where  $T^{-1}$ denotes the inverse function of $T$, $B\in\{0,1\}$ is a Bernoulli random variable with success parameter $\delta_{ab}/(\sigma+\delta_{ab})$, independent of $Z_1$ and $Z_2$, such that $Z_j\sim{\rm Gamma}(\mu+j-1,1/(\mu\sigma))$, $j=1,2$, that is, the density function of $Z_j$ is given by
\begin{align*}
	f_{Z_j}(z)
	=
	{1\over\Gamma(\mu+j-1)[1/(\mu\sigma)]^{\mu+j-1}}\, z^{(\mu+j-1)-1}\exp\left\{-{z\over[1/(\mu\sigma)]}\right\}, \quad z>0, j=1,2.
\end{align*}

Table \ref{table:1-0} \citep{Vila2024b} presents some examples of generators $T(x)$ for use in \eqref{pdf-1} with $a=b$.

\begin{table}[H]
	\caption{Some examples of 
		generators $T(x)$ of probability densities \eqref{pdf-1} with $a= b$.}
	\vspace*{0.15cm}
	\centering 
			\begin{tabular}{lcccccccl} 
				\hline
				Distribution & $\mu$ & $\sigma$ & $T(x)$&
				Parameters
				\\ [0.5ex] 
				\noalign{\hrule height 1.0pt}
				\\ [-1.5ex] 
				Weighted Lindley \citep{Kim2021}
				&  $\phi$  & ${\lambda\over \phi}$   & $x$&
				 $\lambda,\phi>0$
				\\ [2.0ex] 
Weighted inverse  Lindley  &  $\phi$  & ${\lambda\over \phi}$   & ${1\over x}$&
$\lambda,\phi>0$
				\\ [2.0ex] 
New weighted exponentiated Lindley  &  $\phi$  & ${\lambda\over \phi}$   & ${\log(x+1)}$&
$\lambda,\phi>0$
				\\ [2.0ex] 
New weighted log-Lindley  &  $\phi$  & ${\lambda\over \phi}$   & $\exp(x)-1$&
$\lambda,\phi>0$
		\\ [2.0ex]  
		Weighted Nakagami 
		&  $m$  & ${1\over \Omega}$   & $x^2$&
		$m\geqslant {1\over 2}, \ \Omega>0$		
		\\ [2.0ex] 
		Weighted inverse  Nakagami  &   $m$  & ${1\over \Omega}$    & ${1\over x^2}$&
		$m\geqslant {1\over 2}, \ \Omega>0$
		\\ [2.0ex] 
		New weighted exponentiated Nakagami  &   $m$  & ${1\over \Omega}$   & ${\log(x^2+1)}$&
		$m\geqslant {1\over 2}, \ \Omega>0$
		\\ [2.0ex] 
		New weighted log-Nakagami  &   $m$  & ${1\over \Omega}$  & $\exp(x^2)-1$&
		$m\geqslant {1\over 2}, \ \Omega>0$
		\\ [1.0ex] 
		\hline	
\end{tabular}
\label{table:1-0} 
\end{table}

For  $a\neq b$,
Table \ref{table:1} \citep{Vila2024} provides some examples of generators $T(x)$ for use in \eqref{pdf-1}.
\begin{table}[H]
	\caption{Some examples of 
		generators $T(x)$ of exponential family \eqref{pdf-1} with $a\neq b$.}
	\vspace*{0.15cm}
	\centering 
		\renewcommand{\arraystretch}{-0.8}
		\resizebox{\linewidth}{!}{
			\begin{tabular}{lcccrll} 
				\hline
				Distribution & $\mu$ & $\sigma$ & $T(x)$ & Parameters
				\\ [0.5ex] 
				\noalign{\hrule height 1.0pt}
				\\ 
				Nakagami \citep{Laurenson1994}
				&  $m$  & ${1\over \Omega}$   & $x^2$ &  $m\geqslant {1\over 2}$, \
				$\Omega>0$ 
				\\ [2.0ex]	
				Maxwell-Boltzmann \citep{Dunbar1982}
				&  ${3\over 2}$  & ${1\over 3\beta^2}$   & $x^2$ &  $\beta>0$ 
				\\ [2.0ex]	
				Rayleigh \citep{Rayleigh1880}
				&  $1$  & ${1\over 2\beta^2}$   & $x^2$ &  $\beta>0$ 
				\\ [2.0ex]		
				Gamma \citep{Stacy1962}
				&  $\alpha$  &  ${1\over\alpha\beta}$ & $x$ &  $\alpha, \beta > 0$ 
				\\ [2.0ex]		
				Inverse gamma \citep{Cook2008}
				&  $\alpha$  &  ${1\over\alpha\beta}$ & ${1\over x}$ &  $\alpha, \beta > 0$ 
				\\ [2.0ex]
				$\delta$-gamma \citep{Rahman2014}
				&  ${\beta\over \delta}$  &  ${1\over \beta}$ & $x^\delta$ & 
				$\delta, \beta > 0$ 
				\\ [2.0ex]
				Weibull \citep{Johnson1994}		
				&  $1$ & ${1\over \beta^\delta}$ & $x^\delta$ &  
				$\delta, \beta>0$ 
				\\ [2.0ex]		
				Inverse Weibull (Fréchet) \citep{khan2008}
				& $1$ & ${1\over \beta^\delta}$ & ${1\over x^\delta}$ &  
				$\delta, \beta>0$ 
				\\ [2.0ex] 
				Generalized gamma \citep{Stacy1962}
				& ${\alpha\over \delta}$ & ${\delta\over \alpha \beta^\delta}$ & $x^\delta$ &  $\alpha,\delta, \beta>0$ 
				\\ [2.0ex] 
				Generalized inverse gamma \citep{Lee1991}
				&  ${\alpha\over \delta}$ &  ${\delta\over \alpha\beta^{\delta}}$ & ${1\over x^\delta}$ &  
				$\alpha,\delta, \beta>0$ 
				\\ [2.0ex] 
				New log-generalized gamma
				& ${\alpha\over \delta}$ & ${\delta\over \alpha \beta^\delta}$ & $[\exp(x)-1]^\delta $ & $\alpha,\delta>0$,  $\beta>0$
				\\ [2.0ex] 
				New log-generalized inverse gamma
				& ${\alpha\over \delta}$ & ${\delta\over \alpha \beta^\delta}$ & $\big[\exp\big({1\over x}\big)-1\big]^\delta $ & $\alpha,\delta>0$,  $\beta>0$			
				\\ [2.0ex] 
				New exponentiated generalized gamma
				& ${\alpha\over \delta}$ & ${\delta\over \alpha \beta^\delta}$ & $\log^\delta(x+1) $ & $\alpha,\delta>0$,  $\beta>0$
				\\ [2.0ex] 
				New exponentiated generalized inverse gamma
				& ${\alpha\over \delta}$ & ${\delta\over \alpha \beta^\delta}$ & $\log^\delta\big({1\over x}+1\big) $ & $\alpha,\delta>0$,  $\beta>0$
				\\ [2.0ex] 
				New modified log-generalized gamma
				& ${\alpha\over \delta}$ & ${\delta\over \alpha \beta^\delta}$ & $\exp^\delta\big(x-{1\over x}\big) $ & $\alpha,\delta>0$,  $\beta>0$
				\\ [2.0ex] 	
				New extended log-generalized gamma
				& ${\alpha\over \delta}$ & ${\delta\over \alpha \beta^\delta}$ & $x^\delta[\exp(x)-1]^\delta$ & $\alpha,\delta>0$,  $\beta>0$
				\\ [2.0ex] 
				Chi-squared \citep{Johnson1994}	
				& ${\nu\over 2}$ & ${1\over \nu}$ & $x$ & $\nu>0$ 
				\\ [2.0ex] 	
				Scaled inverse chi-squared  \citep{Bernardo1993}
				& ${\nu\over 2}$ & $\tau^2$ &  ${1\over x}$ &  $\nu, \tau^2>0$ 
				\\ [2.0ex] 	
				Gompertz \citep{Gompertz1825}
				& $1$ & $\alpha$ &  $\exp(\delta x)-1$ &  $\alpha, \delta>0$ 
				\\ [2.0ex] 		
				Modified Weibull extension \citep{Xie2022}
				& $\lambda\alpha$ & $1$ &  $\exp\big[ \left({x\over\alpha}\right)^\beta\big]-1$ &  $\alpha,\lambda,\beta>0$
				\\ [2.0ex] 		
				Traditional Weibull \citep{Nadarajah2005}
				& $1$ & $a$ &  $x^b[\exp( cx^d)-1]$ &  $a, d>0$, \ $b,c\geqslant 0$
				\\ [2.0ex] 		
				Flexible Weibull \citep{Bebbington2007}
				& $1$ & $a$ &  $\exp\left(b x-{c\over x}\right)$ &  $a, b, c>0$
				\\ [2.0ex] 		
				Burr type XII (Singh-Maddala) \citep{Burr1942}
				& $1$ & $k$ &  $\log(x^c+1)$ &  $c, k>0$
				\\ [2.0ex] 		
				Dagum (Mielke Beta-Kappa) \citep{Dagum1975}
				& $1$ & $k$ &  $\log\big({1\over x^c}+1\big)$ &  $c, k>0$
				\\ [1.5ex] 	
				\hline	
			\end{tabular}
		}
	\label{table:1} 
\end{table}


There are a number of different methodological proposals in the literature for obtaining close-form estimators; see, for example, the moment-based type \citep{Cheng-Beaulieu2002} and the score-adjusted approaches \citep{Nawa2023, Tamae2020}. Closed-form estimators based on the likelihood function have also been suggested. A type of likelihood-based estimator is obtained by considering the likelihood equations of the generalized distribution obtained by a power transformation and taking the baseline distribution as a special case; see for example, \cite{YCh2016}, \cite{Vila2024}, \cite{RLR2016}, \cite{Kim2021} and \cite{Zhao2021}. For other proposals for likelihood-based estimators, the reader is referred to \cite{Kim2022} and \cite{Cheng-Beaulieu2001}.

This paper develops moment-based closed-form estimators for the parameters of probability distributions of the weighted exponential family \eqref{pdf-1} in the special case where $T(x)=x^{-s}$, $x>0$ and $s\neq 0$. The main motivation for choosing this type of generator is that it provides closed-form expressions for moments (of functions) of $X$ in \eqref{rep-st-1}, which allows finding moment-based estimators for the corresponding parameters. Note that this type of generator includes the $T(x)$ generators of Nakagami, Maxwell-Boltzmann, Rayleigh, gamma, inverse gamma, $\delta$-gamma, Weibull, inverse Weibull, generalized gamma, generalized inverse gamma, chi-squared, scaled inverse chi-squared, weighted Lindley, weighted inverse Lindley, weighted Nakagami and weighted inverse Nakagami; see Tables \ref{table:1-0} and \ref{table:1}.

The rest of the paper is structured as follows. In section \ref{prem_results} we briefly present some preliminary results. In Sections \ref{The New Estimators} and \ref{Asymptotic behavior of estimators}, we describe the newly proposed moment-based estimators and some asymptotic results respectively. In Section \ref{Sec:simulation}, we perform a Monte Carlo simulation study to assess the
performance of a bootstrap bias-reduced version of the proposed estimators.


%
%

\section{Preliminary results}\label{prem_results}
In this section, closed-form expressions for moments (of functions) of $X$ in \eqref{rep-st-1} are provided.

\subsection{Moments}

\begin{proposition}\label{expectation}
If $X$ has density in \eqref{pdf-1}, then, for any measurable function $g:(0,\infty)\to\mathbb{R}$, we have
\begin{align*}
	\mathbb{E}[g(X)]
	=
		{\sigma\over \sigma+\delta_{ab}}\, 
	\mathbb{E}
	[
	g(T^{-1}(Z_1))
	]
	+
	{\delta_{ab}\over \sigma+\delta_{ab}}\, 
	\mathbb{E}
	[
	g(T^{-1}(Z_2))
	].
	\end{align*}
\end{proposition}
\begin{proof}
	Taking into account the characterization \eqref{rep-st-1} of $X$  and by using the independence of $B\sim{\rm Bernoulli}(\delta_{ab}/(\sigma+\delta_{ab}))$ with $Z_1$ and $Z_2$, we obtain
\begin{align*}
	\mathbb{E}[g(X)]
	&=
	\mathbb{E}
	[
	g(
	(1-B)T^{-1}(Z_1)+BT^{-1}(Z_2)
	)
	]
	\nonumber
	\\[0,2cm]
	&=
	\mathbb{E}(1-B)
	\mathbb{E}
	[
	g(T^{-1}(Z_1))
	]
	+
	\mathbb{E}(B)
	\mathbb{E}
	[
	g(T^{-1}(Z_2))
	]
	\nonumber
	\\[0,2cm]
	&=
	{\sigma\over \sigma+\delta_{ab}}\, 
	\mathbb{E}
	[
	g(T^{-1}(Z_1))
	]
	+
	{\delta_{ab}\over \sigma+\delta_{ab}}\, 
	\mathbb{E}
	[
	g(T^{-1}(Z_2))
	].
\end{align*}
This completes the proof.
\end{proof}

\begin{proposition}\label{dist-exp-1}
	If  $T(x)=x^{-s}$, $x>0$ and $s\neq 0$, then the real moments of $X$ are  given by
	\begin{align*}
		\mathbb{E}(X^q)
		=
		(\mu\sigma)^{q\over s} \, {\Gamma(\mu-{q\over s})\over\Gamma(\mu)}
\left[
		{\sigma\over \sigma+\delta_{ab}}
		+
		{\delta_{ab}\over \sigma+\delta_{ab}}\, 
		{(\mu-{q\over s})\over \mu}
\right],
		\quad \text{where} \ \mu-{q\over s}>0.
	\end{align*}
	%
\end{proposition}
\begin{proof}
By using Proposition \ref{expectation}, with $g(x)=x^q$ and $T(x)=x^{-s}$, $x>0$ and $s\neq 0$, we have	
\begin{align}\label{id-pre-1}
	\mathbb{E}(X^q)
	=
	{\sigma\over \sigma+\delta_{ab}}\, 
	\mathbb{E}
	(
	Z_1^{-q/s}
	)
	+
	{\delta_{ab}\over \sigma+\delta_{ab}}\, 
	\mathbb{E}
	(
	Z_2^{-q/s}
	).
\end{align}
As $Z_j\sim{\rm Gamma}(\mu+j-1,1/(\mu\sigma))$, $j=1,2$ and
\begin{align*}
\mathbb{E}(Z^\nu)=\theta^{\nu}\, {\Gamma(k+\nu)\over\Gamma(k)},
\quad 
Z\sim {\rm Gamma}(k,\theta),  \
\nu>-k,
\end{align*}
from \eqref{id-pre-1}, the proof follows.
\end{proof}

\begin{proposition}\label{dist-exp-2}
		If  $T(x)=x^{-s}$, $x>0$ and $s\neq 0$, then
	\begin{multline*}
	\mathbb{E}\left[{X^{-p}\log(X)}\right]
	=
	-{1\over s}\,{\Gamma(\mu+{p\over s})\over (\mu\sigma)^{p/s}\Gamma(\mu)}
\Bigg\{
{\sigma\over \sigma+\delta_{ab}}
\left[
\psi^{(0)}\left(\mu+{p\over s}\right)-\log(\mu\sigma)
\right]
	\\[0,2cm]
+
{\delta_{ab}\over \sigma+\delta_{ab}}\, 
{\mu+{p\over s}\over\mu}
\left[
\psi^{(0)}\left(\mu+{p\over s}+1\right)-\log(\mu\sigma)
\right]
\Bigg\},
\quad
\text{where} \ \mu+{p\over s}>0.
	\end{multline*}
\end{proposition}
\begin{proof}
By using Proposition \ref{expectation}, with $g(x)=\log(x)/x^p$ and $T(x)=x^{-s}$, $x>0$ and $s\neq 0$, we have	
\begin{align}\label{id-pre-2}
		\mathbb{E}\left[{X^{-p}\log(X)}\right]
	=
	-{1\over s}
	\left\{
	{\sigma\over \sigma+\delta_{ab}}\, 
	\mathbb{E}
	\left[{Z_1^{p/s}\log(Z_1)}\right]
	+
	{\delta_{ab}\over \sigma+\delta_{ab}}\, 
	\mathbb{E}
	\left[{Z_2^{p/s}\log(Z_2)}\right]
	\right\}.
\end{align}
As $Z_j\sim{\rm Gamma}(\mu+j-1,1/(\mu\sigma))$, $j=1,2$ and
\begin{align}\label{exp-log}
\mathbb{E}\left[Z^{\nu}\log(Z)\right]
=
{\Gamma(k+\nu)\over \theta^{-\nu}\Gamma(k)}\,
[\psi^{(0)}(k+\nu)-\log(1/\theta)],
\quad 
Z\sim {\rm Gamma}(k,\theta), \ 
k>-\nu,
\end{align}
from \eqref{id-pre-2}, the proof follows.
\end{proof}

\begin{proposition}\label{dist-exp-3}
	If  $T(x)=x^{-s}$, $x>0$ and $s\neq 0$, then
	\begin{align*}
		\mathbb{E}\left[{X^{-s}\log(X^{-s})\over 1+\delta_{ab}X^{-s}}\right]
		=
		{1\over \sigma+\delta_{ab}}\,
		[\psi^{(0)}(\mu+1)-\log(\mu\sigma)].
	\end{align*}		
	\end{proposition}
	\begin{proof}
By using the definition \eqref{pdf-1} with $T(x)=x^{-s}$, $x>0$ and $s\neq 0$, and by making the change of variable $y=x^{-s}$, note that
	\begin{align*}
		\mathbb{E}\left[{X^{-s}\log(X^{-s})\over 1+\delta_{ab}X^{-s}}\right]
	&=
	{s(\mu\sigma)^{\mu+1} \over (\sigma+\delta_{ab})\Gamma(\mu+1)}
	\int_{0}^\infty
	\log(x^{-s})\,
	(x^{-s})^{\mu+{1\over s}+1}
	\exp\left\{-\mu \sigma x^{-s}\right\}
	{\rm d}x
	\\[0,2cm]
	&=
	{(\mu\sigma)^{\mu+1} \over (\sigma+\delta_{ab})\Gamma(\mu+1)}
\int_{0}^\infty
\log(y) y^{\mu}
\exp\left\{-\mu \sigma y\right\}
{\rm d}y
	\\[0,2cm]
&=
	{1\over \sigma+\delta_{ab}}\, \mathbb{E}[\log(Y)],
	\end{align*}
	where $Y\sim {\rm Gamma}(\mu+1,1/(\mu\sigma))$.
	By employing formula in \eqref{exp-log} with $\nu=0$, $k=\mu+1$ and $\theta=1/(\mu\sigma)$, in the last equality,
	the proof readily follows.
	\end{proof}

\subsection{Main formulas}
By using Proposition \ref{dist-exp-1} with $q=-s$, we have
	\begin{align}\label{1-identity}
	\mathbb{E}\left({X^{-s}}\right)
	=
{1\over \sigma}
\left(
1
+
{\delta_{ab}\over \sigma+\delta_{ab}}\,
{1\over \mu}
\right).
\end{align}

By using Proposition \ref{dist-exp-2} with $p=0$, we have
\begin{align}\label{expect-log}
	\mathbb{E}[\log(X)]
	&=
	-{1\over s}
	\left\{
	{\sigma\over \sigma+\delta_{ab}}\, 
	[
	\psi^{(0)}(\mu)-\log(\mu\sigma)
	]
	+
	{\delta_{ab}\over \sigma+\delta_{ab}}\, 
	[
	\psi^{(0)}(\mu+1)-\log(\mu\sigma)
	]
	\right\}
	\nonumber
		\\[0,2cm]
	&=
	-{1\over s}
	\left[
	\psi^{(0)}(\mu+1)-\log(\mu\sigma)
	-
	{1\over\mu}
	\left(
	{\sigma\over \sigma+\delta_{ab}}
	\right)
	\right],
\end{align}
where in the last line the identity $\psi^{(0)}(z+1)=\psi^{(0)}(z)+1/z$ has been used.

By using Proposition \ref{dist-exp-2} with $p=s$, we have
\begin{align*}
\mathbb{E}\left[{X^{-s}\log(X)}\right]
&=
-{1\over s}\,{1\over \sigma}
\Bigg\{
{\sigma\over \sigma+\delta_{ab}}\,
[
\psi^{(0)}(\mu+1)-\log(\mu\sigma)
]
+
{\delta_{ab}\over \sigma+\delta_{ab}}\, 
{\mu+1\over\mu}\,
[
\psi^{(0)}(\mu+2)-\log(\mu\sigma)
]
\Bigg\}
\\[0,2cm]
&=
{1\over \sigma}
\Bigg\{
-{1\over s}
\left[
\psi^{(0)}(\mu+1)-\log(\mu\sigma)
	-
{1\over\mu}
\left(
{\sigma\over \sigma+\delta_{ab}}
\right)
\right]
\\[0,2cm]
&
\hspace*{6.5cm}
-
{1\over s\mu}
\left\{
{\delta_{ab}\over \sigma+\delta_{ab}}\, 
[
\psi^{(0)}(\mu+1)-\log(\mu\sigma)
]
+
1
\right\}
\Bigg\},
\end{align*}
where in the last equality we have again used the identity $\psi^{(0)}(z+1)=\psi^{(0)}(z)+1/z$. By \eqref{expect-log} and Proposition \ref{dist-exp-3} note that $\mathbb{E}\left[{X^{-s}\log(X)}\right]$ can be written as
\begin{align*}
\mathbb{E}\left[{X^{-s}\log(X)}\right]
=
{1\over \sigma}
\Bigg\{
\mathbb{E}[\log(X)]
-
{1\over s\mu}
\left\{
{\delta_{ab}}
\mathbb{E}\left[{X^{-s}\log(X^{-s})\over 1+\delta_{ab}X^{-s}}\right]
+
1
\right\}
\Bigg\},
\end{align*}
from which we can express $\mu$ as follows:
\begin{align}\label{2-identity}
	\mu
	=
	\dfrac{\displaystyle
		{\delta_{ab}}
		\mathbb{E}\left[h_1(X)\right]
		+
		1
		}
		{\displaystyle
			\sigma \mathbb{E}\left[h_2(X)\right]
			-
			\mathbb{E}[h_3(X)]
			},
\end{align}
where we have adopted the following notations:
\begin{align}\label{def-h}
	h_1(x)\equiv {x^{-s}\log(x^{-s})\over 1+\delta_{ab}x^{-s}},
	\quad
	h_2(x)\equiv {x^{-s}\log(x^{-s})},
\quad
	h_3(x)\equiv \log(x^{-s}).
\end{align}

Plugging \eqref{2-identity} into \eqref{1-identity} and solving for $\sigma$ gives
\begin{align}\label{sigma-function}
	\sigma
	&=
	\dfrac{
	1
	-
	\delta_{ab}\mathbb{E}\left[h_4(X)\right]
	+
	{\delta_{ab}\mathbb{E}\left[h_2(X)\right]\over 	
		\delta_{ab}
		\mathbb{E}\left[h_1(X)\right]
		+
		1}
	}
		{
		2\mathbb{E}\left[h_4(X)\right]
	}
	\nonumber
		\\[0,2cm]
		&+
			\dfrac{
		\sqrt{\left\{	1
			-
			\delta_{ab}\mathbb{E}\left[h_4(X)\right]
			+
			{\delta_{ab}\mathbb{E}\left[h_2(X)\right]\over 	
				\delta_{ab}
				\mathbb{E}\left[h_1(X)\right]
				+
				1}\right\}^2
				+
			4\mathbb{E}\left[h_4(X)\right]
			\left\{\delta_{ab}-{\delta_{ab}\mathbb{E}\left[h_3(X)\right]\over
			\delta_{ab}
			\mathbb{E}\left[h_1(X)\right]
			+
			1 }\right\}
			}
	}
	{
			2\mathbb{E}\left[h_4(X)\right]
	},
\end{align}
where
\begin{align}\label{def-h-alt}
	h_4(x)\equiv T(x)= x^{-s}.
\end{align}

\section{Closed-form estimators}\label{The New Estimators}

Let $\{X_i : i = 1,\ldots , n\}$ be a univariate random sample of size $n$ from  $X$ having density in \eqref{pdf-1}. 

By using the method of moments in \eqref{sigma-function}, the corresponding sample moment to obtain the estimator of $\sigma$ is
\begin{align}\label{sigma-estimator-1}
	\widehat{\sigma}
	=
	\dfrac{
		1
		-
		\delta_{ab}\overline{X}_4
		+
		{\delta_{ab}\overline{X}_2\over 	
			\delta_{ab}
			\overline{X}_1
			+
			1}
		+
		\sqrt{\left\{	1
			-
			\delta_{ab}\overline{X}_4
			+
			{\delta_{ab}\overline{X}_2\over 	
				\delta_{ab}
				\overline{X}_1
				+
				1}\right\}^2
			+
			4\overline{X}_4
			\left\{\delta_{ab}-{\delta_{ab}\overline{X}_3\over
				\delta_{ab}
				\overline{X}_1
				+
				1 }\right\}
		}
	}
	{
		2\overline{X}_4
	},
\end{align}
where we have defined
\begin{align}\label{def-mean-aritm}
	\overline{\boldsymbol{X}}
	\equiv
		\begin{pmatrix}
		\overline{X}_1\\[0,2cm]
		\overline{X}_2\\[0,2cm]
		\overline{X}_3\\[0,2cm]
		\overline{X}_4
	\end{pmatrix}
	=
	{1\over n}
	\sum_{i=1}^{n}
	\begin{pmatrix}
	h_1(X_i)\\[0,2cm]
		h_2(X_i)\\[0,2cm]
			h_3(X_i)\\[0,2cm]
				h_4(X_i)
	\end{pmatrix},
\end{align}
with $h_1,h_2,h_3$ and $h_4$ being as in \eqref{def-h} and \eqref{def-h-alt}, respectively.

Plugging \eqref{sigma-estimator-1} in \eqref{2-identity} and using the method of moments, the sample moment to obtain the estimator of $\mu$ is
\begin{align}\label{mu-estimator}
	\widehat{\mu}
	=
	\dfrac{\displaystyle
		{\delta_{ab}}
		\overline{X}_1
		+
		1
	}
	{\displaystyle
		\widehat{\sigma} \overline{X}_2
		-
		\overline{X}_3
	}.
\end{align}

\subsection{Case $a\neq b$}
In this case, from \eqref{sigma-estimator-1} and \eqref{mu-estimator}, we obtain the new estimators for $\sigma$ and $\mu$ as
\begin{align}\label{sigma-estimator-1-1}
	\widehat{\sigma}
	=
	\dfrac{
1
	}
	{
		\overline{X}_4
	}
\end{align}
and 
\begin{align}\label{mu-estimator-1}
	\widehat{\mu}
	=
	\dfrac{
		1
	}
	{\displaystyle
		\widehat{\sigma} \overline{X}_2
		-
		\overline{X}_3
	}.
\end{align}
respectively, where $\overline{X}_2, \overline{X}_3$ and $\overline{X}_4$ are as given in \eqref{def-mean-aritm}. Note that, in this case, $\widehat{\sigma}$ coincides with the maximum likelihood (ML) estimator \citep[see Equation 6 in][]{Vila2024}.

\subsection{Case $a= b$}

In this case, from \eqref{sigma-estimator-1} and \eqref{mu-estimator}, we obtain the new estimators for $\sigma$ and $\mu$ as
\begin{align}\label{sigma-estimator-1-0}
	\widehat{\sigma}
	=
	\dfrac{
		1
		-
		\overline{X}_4
		+
		{\overline{X}_2\over 	
			\overline{X}_1
			+
			1}
		+
		\sqrt{\left\{	1
			-
			\overline{X}_4
			+
			{\overline{X}_2\over 	
				\overline{X}_1
				+
				1}\right\}^2
			+
			4\overline{X}_4
			\left\{1-{\overline{X}_3\over
				\overline{X}_1
				+
				1 }\right\}
		}
	}
	{
		2\overline{X}_4
	}
\end{align}
and
\begin{align}\label{mu-estimator-0}
	\widehat{\mu}
	=
	\dfrac{\displaystyle
		\overline{X}_1
		+
		1
	}
	{\displaystyle
		\widehat{\sigma} \overline{X}_2
		-
		\overline{X}_3
	},
\end{align}
respectively, where $\overline{X}_1$ (with $a=b$), $\overline{X}_2, \overline{X}_3$ and $\overline{X}_4$ are as given in \eqref{def-mean-aritm}.

\begin{remark}
It is simple to observe that the weighted probability distributions in \eqref{pdf-1} belongs to the exponential family with vector of sufficient statistics given by $(T(x),\log(T(x)))^\top$. Furthermore, for exponential families with sufficient statistics  $(T(x),\log(T(x)))^\top$ it is well-known that the moment-type  estimators are in fact the maximum likelihood estimators \citep{Davidson1974}. The estimators $\widehat{\mu}$ and $\widehat{\sigma}$ (in \eqref{sigma-estimator-1-1} and \eqref{mu-estimator-1} for case $a\neq b$, and \eqref{sigma-estimator-1-0} and \eqref{mu-estimator-0} for case $a= b$) were derived using the moment-type method (modification of the method of moments) with vector of statistics
\begin{align*}
\left(T(x),\log(T(x)),T(x)\log(T(x)),{T(x)\log(T(x))\over 1+\delta_{ab}T(x)}\right)^\top,
\quad
\text{for}
\ 
T(x)=x^{-s}, \, x>0 \, \text{and} \, s\neq 0.
\end{align*}
By using the maximum likelihood equations corresponding to an extension of the weighted probability model in \eqref{pdf-1}, simple closed-forms for the estimators of $\mu$ and $\sigma$, denoted by $\widehat{\mu}_{\bullet}$ and $\widehat{\sigma}_{\bullet}$, respectively, were derived in reference \cite{Vila2024} (case $a\neq b$) and \cite{Vila2024b} (case $a= b$). Note that these estimators are not, in fact, the maximum likelihood estimators corresponding to the probability distribution in \eqref{pdf-1}, so we cannot say with certainty that the estimators provided by the moment-type method, $\widehat{\mu}$ and $\widehat{\sigma}$, are the same as $\widehat{\mu}_{\bullet}$ and $\widehat{\sigma}_{\bullet}$ in the special case $T(x)=x^{-s}$, $x>0$ and $s\neq 0$. In other words, we cannot apply the results obtained in \cite{Davidson1974}. However, to our surprise, through a simple but laborious calculation and some  simulations, we found that estimators $\widehat{\mu}$ and $\widehat{\sigma}$ (in \eqref{sigma-estimator-1-1} and \eqref{mu-estimator-1} for case $a\neq b$, and \eqref{sigma-estimator-1-0} and \eqref{mu-estimator-0} for case $a= b$) coincide with the estimators $\widehat{\mu}_{\bullet}$ and $\widehat{\sigma}_{\bullet}$ obtained in \cite{Vila2024} (case $a\neq b$) and \cite{Vila2024b} (case $a= b$).
\end{remark}

\section{Asymptotic behavior of estimators}\label{Asymptotic behavior of estimators}

Let $X_1,\ldots, X_n$ be a random sample of size $n$ from the  variable $X$ with PDF \eqref{pdf-1}. If we further let  $\overline{\boldsymbol{X}}=(\overline{X}_1,\overline{X}_2,\overline{X}_3,\overline{X}_4)^\top$ and 
\begin{align*}
\boldsymbol{X}
\equiv
\begin{pmatrix}
	h_1(X)\\[0,2cm]
	h_2(X)\\[0,2cm]
	h_3(X)\\[0,2cm]
	h_4(X)
\end{pmatrix},
\end{align*}
with $\overline{X}_i$, $i=1,\ldots,6$, as given in \eqref{def-mean-aritm}, and  $h_1,h_2,h_3$ and $h_4$ being as in \eqref{def-h} and \eqref{def-h-alt}, respectively.
By applying strong law of large numbers, we have
\begin{align*}
	\overline{\boldsymbol{X}}\stackrel{\rm a.s.}{\longrightarrow}
	\mathbb{E}(\boldsymbol{X}),
\end{align*}
where ``$\stackrel{\rm a.s.}{\longrightarrow}$'' denotes almost sure convergence.
Hence, continuous-mapping theorem \citep{Billingsley1969}  gives
\begin{align}\label{id-1}
\widehat{\sigma}=g_1(\overline{\boldsymbol{X}})\stackrel{\rm a.s.}{\longrightarrow}
g_1(\mathbb{E}(\boldsymbol{X}))
\quad \text{and} \quad
\widehat{\mu}=g_2(\overline{\boldsymbol{X}})\stackrel{\rm a.s.}{\longrightarrow}
g_2(\mathbb{E}(\boldsymbol{X})),
\end{align}
with
\begin{align}\label{def-g1}
g_1(x_1,x_2,x_3,x_4)
\equiv
	\dfrac{
	1
	-
	\delta_{ab}x_4
	+
	{\delta_{ab}x_2\over 	
		\delta_{ab}
		x_1
		+
		1}
	+
	\sqrt{\left\{	1
		-
		\delta_{ab}x_4
		+
		{\delta_{ab}x\over 	
			\delta_{ab}
			x_1
			+
			1}\right\}^2
		+
		4x_4
		\left\{\delta_{ab}-{\delta_{ab}x_3\over
			\delta_{ab}
			x_1
			+
			1 }\right\}
	}
}
{
	2x_4
}
\end{align}
and
\begin{align}\label{def-g2}
	g_2(x_1,x_2,x_3,x_4)
\equiv
	\dfrac{\displaystyle
	{\delta_{ab}}
	x_1
	+
	1
}
{\displaystyle
	g_1(x_1,x_2,x_3,x_4) x_2
	-
	x_3
}.
\end{align}

Furthermore, by Central limit theorem,
\begin{align*}
	\sqrt{n}\big[\overline{\boldsymbol{X}}-\mathbb{E}(\boldsymbol{X})\big]\stackrel{\mathscr D}{\longrightarrow} N_4(\bm 0, \bm\Sigma),
\end{align*}
where $\bm\Sigma$ denotes the covariance matrix of $\boldsymbol{X}$ and ``$\stackrel{\mathscr D}{\longrightarrow}$'' means convergence in distribution.
So, delta method provides
\begin{align}\label{id-2}
	\sqrt{n}
	\left[
	\begin{pmatrix}
	\widehat{\mu}
	\\[0,2cm]
	\widehat{\sigma}
	\end{pmatrix}
	-
	\begin{pmatrix}
	g_2(\mathbb{E}(\bm X))
	\\[0,2cm]
	g_1(\mathbb{E}(\bm X))
	\end{pmatrix}
	\right]
	\stackrel{\eqref{id-1}}{=}
		\sqrt{n}
	\left[
	\begin{pmatrix}
		g_2(\overline{\boldsymbol{X}})
		\\[0,2cm]
		g_1(\overline{\boldsymbol{X}})
	\end{pmatrix}
	-
	\begin{pmatrix}
		g_2(\mathbb{E}(\bm X))
		\\[0,2cm]
		g_1(\mathbb{E}(\bm X))
	\end{pmatrix}
	\right]
	\stackrel{\mathscr D}{\longrightarrow}
	N_2(\bm 0, \bm A\bm\Sigma \bm A^\top),
\end{align}
with $\bm A$ being the partial derivatives matrix defined as
\begin{align*}
	\bm A
	=
	\left. 
	\begin{pmatrix}
	\displaystyle
	{\partial g_2(\bm x)\over\partial x_1} & \displaystyle {\partial g_2(\bm x)\over\partial x_2} & \displaystyle {\partial g_2(\bm x)\over\partial x_3} & \displaystyle {\partial g_2(\bm x)\over\partial x_4} 
	\\[0,5cm]
	\displaystyle
	{\partial g_1(\bm x)\over\partial x_1} & \displaystyle {\partial g_1(\bm x)\over\partial x_2} & \displaystyle {\partial g_1(\bm x)\over\partial x_3} & \displaystyle {\partial g_1(\bm x)\over\partial x_4} 
	\end{pmatrix}\,
	\right\vert_{\bm x=\mathbb{E}(\bm X)}.
\end{align*}
For simplicity of presentation, we do not present the partial derivatives of $g_j$, $j=1,2$, here. Analogously to the calculation of $\mathbb{E}(\bm X)$, the second moments of the components of $\bm X$ can be determined which is sufficient to guarantee the existence of the matrix $\bm\Sigma$.

The following result shows that for generators of type $T(x)=x^{-s}$, $x>0$ and $s\ne 0$, the strong consistency property and a Central limit type theorem for the estimators $\widehat{\sigma}$ and $\widehat{\mu}$ are satisfied. 

\begin{theorem}
	If $T(x)=x^{-s}$, $x>0$ and $s\ne 0$, then
	$g_1(\mathbb{E}(\bm X))=\sigma$ and $g_2(\mathbb{E}(\bm X))=\mu$, where $g_1$ and $g_2$ are given in \eqref{def-g1} and \eqref{def-g2}, respectively. Moreover, from \eqref{id-2},
\begin{align*}
	\sqrt{n}
	\left[
	\begin{pmatrix}
		\widehat{\mu}
		\\[0,2cm]
		\widehat{\sigma}
	\end{pmatrix}
	-
	\begin{pmatrix}
		\mu
		\\[0,2cm]
		\sigma
	\end{pmatrix}
	\right]
	\stackrel{\mathscr D}{\longrightarrow}
	N_2(\bm 0, \bm A\bm\Sigma \bm A^\top),
\end{align*}	
	where $\bm A$ was given lines above and $\bm\Sigma$ is the covariance matrix of $\boldsymbol{X}$.
\end{theorem}
\begin{proof}
	The proof follows immediately from \eqref{2-identity} and \eqref{sigma-function}, because
	\begin{align*}
	g_2(\mathbb{E}(\bm X))
	=
	\dfrac{\displaystyle
		{\delta_{ab}}
		\mathbb{E}\left[h_1(X)\right]
		+
		1
	}
	{\displaystyle
		\sigma \mathbb{E}\left[h_2(X)\right]
		-
		\mathbb{E}[h_3(X)]
	}
	\stackrel{\eqref{2-identity}}{=}
	\mu
	\end{align*}
	and
		\begin{align*}
		g_1(\mathbb{E}(\bm X))
		&=
			\dfrac{
			1
			-
			\delta_{ab}\mathbb{E}\left[h_4(X)\right]
			+
			{\delta_{ab}\mathbb{E}\left[h_2(X)\right]\over 	
				\delta_{ab}
				\mathbb{E}\left[h_1(X)\right]
				+
				1}
			}
					{
				2\mathbb{E}\left[h_4(X)\right]
			}
			\\[0,2cm]
			&+
			\dfrac{
			\sqrt{\left\{	1
				-
				\delta_{ab}\mathbb{E}\left[h_4(X)\right]
				+
				{\delta_{ab}\mathbb{E}\left[h_2(X)\right]\over 	
					\delta_{ab}
					\mathbb{E}\left[h_1(X)\right]
					+
					1}\right\}^2
				+
				4\mathbb{E}\left[h_4(X)\right]
				\left\{\delta_{ab}-{\delta_{ab}\mathbb{E}\left[h_3(X)\right]\over
					\delta_{ab}
					\mathbb{E}\left[h_1(X)\right]
					+
					1 }\right\}
			}
		}
		{
			2\mathbb{E}\left[h_4(X)\right]
		}
		\stackrel{\eqref{sigma-function}}{=}
		\sigma.
	\end{align*}	

Thus completes the proof.
\end{proof}

\section{Simulation study}\label{Sec:simulation}

In this section, we carry out a Monte Carlo simulation study for evaluating the performance of the proposed estimators. Particularly, we evaluate a bias-reduced version of the proposed moment-type estimators, as they are biased \citep{RLR2016}. For illustrative purposes, we only present the results for the weighted inverse Lindley distribution. Then, by considering the parameters $\mu=\phi$, $\sigma=\lambda/\phi$ and generator $T(x)=x^{-1}$ of the weighted inverse Lindley distribution, given in Table \ref{table:1-0}, from \eqref{sigma-estimator-1-0} and \eqref{mu-estimator-0} the bootstrap biased-reduced moment-type estimators for $\lambda$ and $\phi$ are given by

\begin{eqnarray*}\label{biasred}
\widehat{\lambda}^{*} = 2\widehat{\lambda} - \frac{1}{B}\sum_{b=1}^{B}\widehat{\lambda}^{(b)},\\
\widehat{\phi}^{*} = 2\widehat{\phi} - \frac{1}{B}\sum_{b=1}^{B}\widehat{\phi}^{(b)},\\
\end{eqnarray*}
where $\widehat{\lambda}^{(b)}$ and $\widehat{\phi}^{(b)}$ are the $b$-th bootstrap replicates from the $b$-th bootstrap sample,
\begin{align*}
		\widehat{\lambda}
		=
		\dfrac{
			1
			-
			\overline{X}_4
			+
			{\overline{X}_2\over
				\overline{X}_1
				+
				1}
			+
			\sqrt{\left\{	1
				-
				\overline{X}_4
				+
				{\overline{X}_2\over
					\overline{X}_1
					+
					1}\right\}^2
				+
				4\overline{X}_4
				\left\{1-{\overline{X}_3\over
					\overline{X}_1
					+
					1 }\right\}
			}
		}
		{
			2\overline{X}_4
		}
		\,
		\widehat{\phi}
	\end{align*}
	and
	\begin{align*}
		\widehat{\phi}
		=
		\dfrac{\displaystyle
			\overline{X}_1
			+
			1
		}
		{\displaystyle
				\dfrac{
			1
			-
			\overline{X}_4
			+
			{\overline{X}_2\over
				\overline{X}_1
				+
				1}
			+
			\sqrt{\left\{	1
				-
				\overline{X}_4
				+
				{\overline{X}_2\over
					\overline{X}_1
					+
					1}\right\}^2
				+
				4\overline{X}_4
				\left\{1-{\overline{X}_3\over
					\overline{X}_1
					+
					1 }\right\}
			}
		}
		{
			2\overline{X}_4
		}\, \overline{X}_2
			-
			\overline{X}_3
		},
	\end{align*}
with
	\begin{align*}
	\overline{X}_1
	&=
	{1\over n}\sum_{i=1}^{n}{X_i^{-1}\log(X_i^{-1})\over 1+X_i^{-1}},
	\\[0,2cm]
	\overline{X}_2&={1\over n}\sum_{i=1}^{n}{X_i^{-1}\log(X_i^{-1})},
\\[0,2cm]
	\overline{X}_3&={1\over n}\sum_{i=1}^{n}\log(X_i^{-1}),
\\[0,2cm]
	\overline{X}_4&={1\over n}\sum_{i=1}^{n}X_i^{-1}.
	\end{align*}

For assessing the performance of the proposed bootstrap biased-reduced moment-type estimators, we calculated the relative bias (RB) and root mean square error (RMSE), given by
\begin{eqnarray*}
 \widehat{\textrm{RB}}(\widehat{\theta}^{*}) =  \left|\frac{\frac{1}{N} \sum_{i = 1}^{N} \widehat{\theta}^{*(i)} - \theta}{\theta}\right| ,  \quad
\widehat{\mathrm{RMSE}}(\widehat{\theta}^{*}) = {\sqrt{\frac{1}{N} \sum_{i = 1}^{N} (\widehat{\theta}^{*(i)} - \theta)^2}},
\end{eqnarray*}
where $\theta\in\{ \lambda,\phi\}$ and $\widehat{\theta}^{*(i)}\in\{\widehat{\lambda}^{*(i)},\widehat{\phi}^{*(i)}\}$ are the true parameter value and its $i$-th bootstrap bias-reduced estimate, and $N$ is the number of Monte Carlo replications.

The simulation scenario considers the following setting: sample size $n \in \{20,50,100,200,400,1000\}$, $\phi\in \{0.5,1,3,5,9\}$, and $\lambda=1$.  The number of Monte Carlo replications was $N=1,000$ and the number of bootstrap replications was $B=200$. The numerical evaluations were implemented using the \verb+R+ software; see \url{http://cran.r-project.org}.

Figures~\ref{fig_dagum_MC1} and \ref{fig_dagum_MC2} show the results of Monte Carlo simulation study to assess the performance of the proposed bootstrap biased-reduced moment-type (MOM) estimators.  For comparison purposes, we also considered the results of the classical maximum likelihood estimator (MLE) discussed by \cite{RLSL2018}. Figures~\ref{fig_dagum_MC1} and \ref{fig_dagum_MC2} show that, as expected, both biases and RMSEs of the estimators approach zero as the sample size increases. However, the bias is much lower for smaller samples for the proposed estimator. Moreover, the RMSE is similar for both estimators. Finally, the results do not seem to be affected by the parameter $\phi$.

\begin{figure}[htb!]
	\vspace{-0.25cm}
	\centering
	{\includegraphics[height=6.5cm,width=6.5cm]{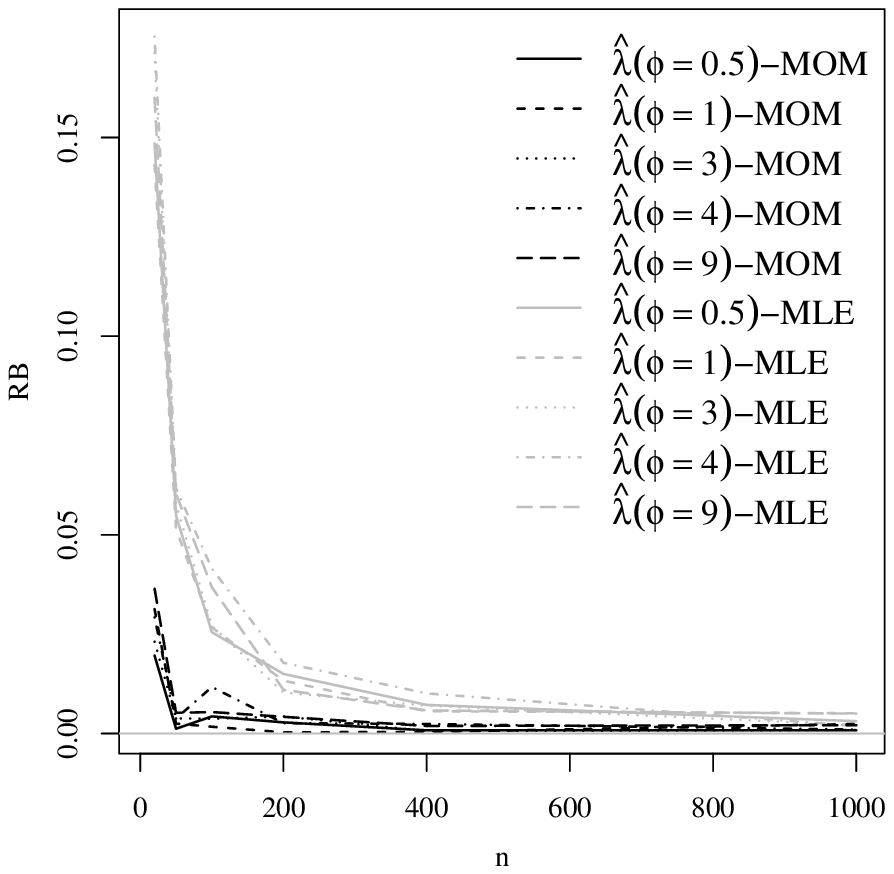}}\hspace{-0.25cm}
	{\includegraphics[height=6.5cm,width=6.5cm]{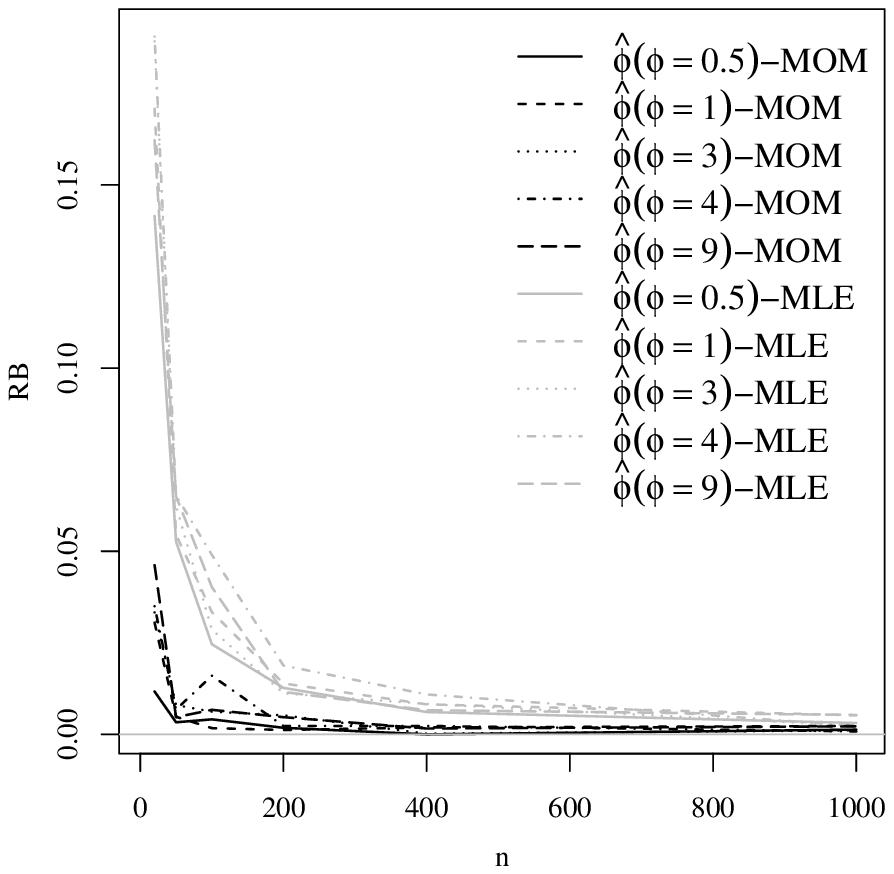}}\hspace{-0.25cm}
	\vspace{-0.2cm}
	\caption{Empirical RB of the bootstrap bootstrap biased-reduced moment-type and classical maximum likelihood estimators estimators for the weighted inverse Lindley distribution.}
	\label{fig_dagum_MC1}
\end{figure}

\begin{figure}[htb!]
	\vspace{-0.25cm}
	\centering
	{\includegraphics[height=6.5cm,width=6.5cm]{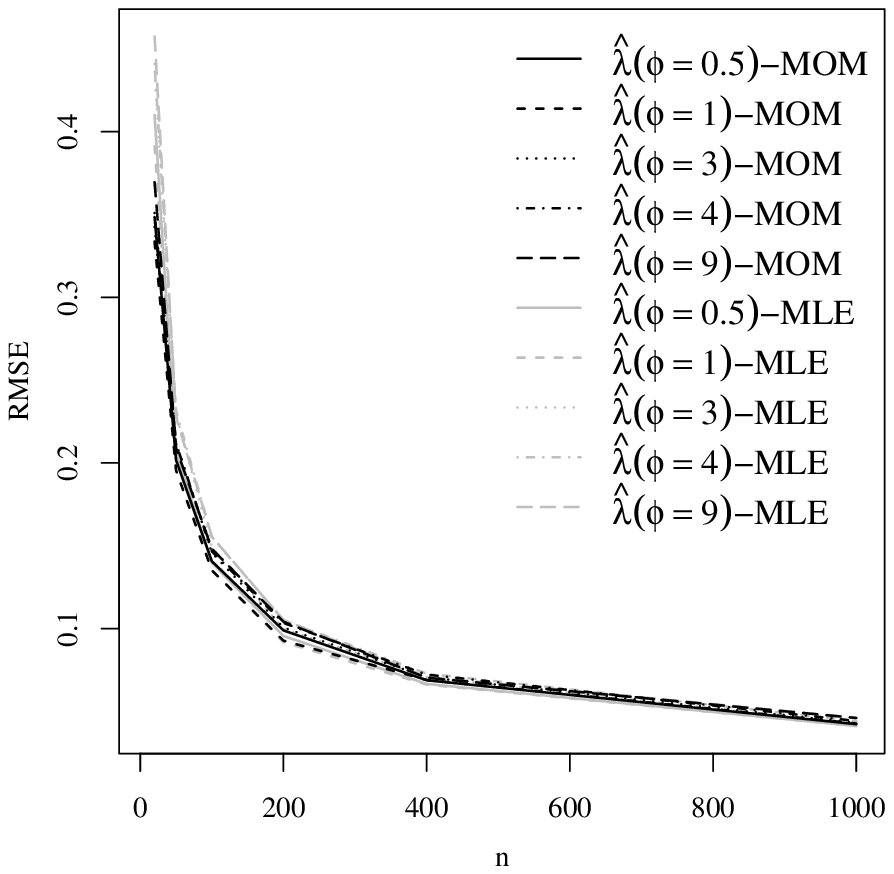}}\hspace{-0.25cm}
	{\includegraphics[height=6.5cm,width=6.5cm]{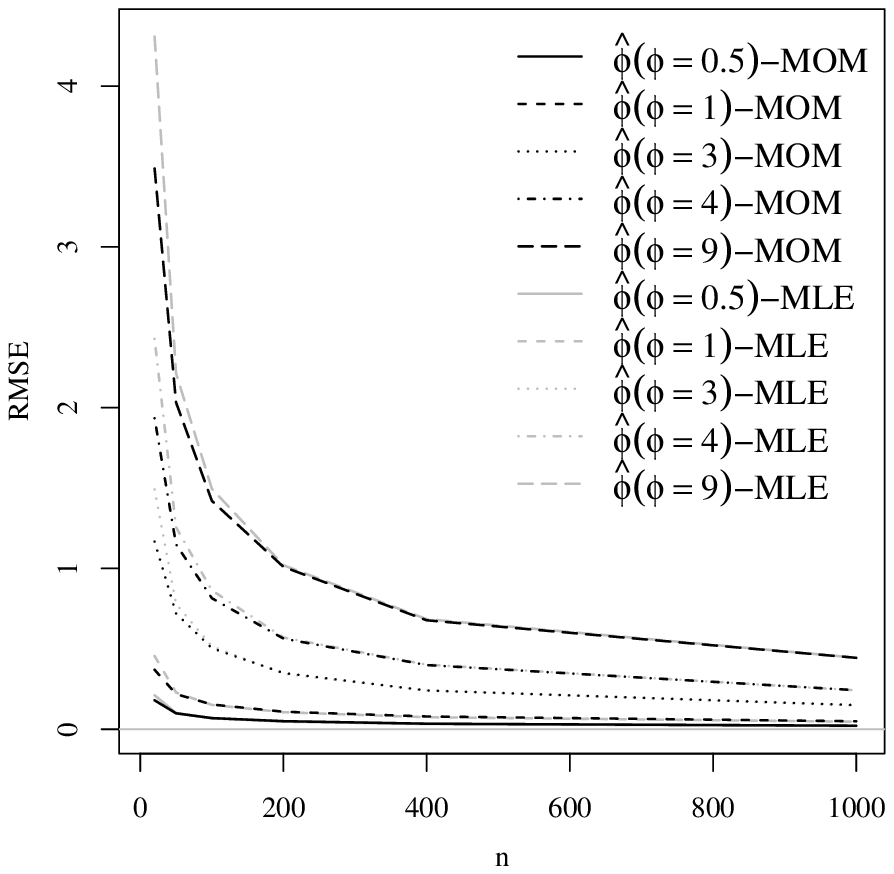}}\hspace{-0.25cm}

	%
	\vspace{-0.2cm}
	\caption{Empirical RMSE of the bootstrap biased-reduced moment-type and classical maximum likelihood estimators for the weighted inverse Lindley distribution.}
	\label{fig_dagum_MC2}
\end{figure}

\clearpage

	\paragraph*{Acknowledgements}
The research was supported in part by CNPq and CAPES grants from the Brazilian government.
	
	\paragraph*{Disclosure statement}
	There are no conflicts of interest to disclose.

\end{document}